\documentclass[lotsofwhite]{patmorin}
\usepackage{amsmath,graphicx}
\usepackage{amsthm}
\usepackage{amsfonts}

\newcommand{\comment}[1]{}

\newcommand{\seclabel}[1]{\label{sec:#1}}

\newcommand{\secref}[1]{\mbox{Section~\ref{sec:#1}}}

\newcommand{\figlabel}[1]{\label{fig:#1}}

\newcommand{\figref}[1]{\mbox{Figure~\ref{fig:#1}}}

\newcommand{\eqlabel}[1]{\label{eq:#1}}
\renewcommand{\eqref}[1]{(\ref{eq:#1})}

\newtheorem{thm}{Theorem}{\bfseries}{\itshape}
\newcommand{\thmlabel}[1]{\label{thm:#1}}
\newcommand{\thmref}[1]{Theorem~\ref{thm:#1}}

\newtheorem{lem}{Lemma}{\bfseries}{\itshape}
\newcommand{\lemlabel}[1]{\label{lem:#1}}
\newcommand{\lemref}[1]{Lemma~\ref{lem:#1}}

\newtheorem{cor}{Corollary}{\bfseries}{\itshape}
\newcommand{\corlabel}[1]{\label{cor:#1}}
\newcommand{\corref}[1]{Corollary~\ref{cor:#1}}

{\bfseries}{\itshape}

{\bfseries}{\itshape}

\newtheorem{assumption}{Assumption}{\bfseries}{\rm}

\newcommand{\R}{\mathbb{R}}

\newcommand{\Sp}{\mathbb{S}}

\usepackage{marvosym}

\title{\MakeUppercase{Algorithms for Marketing-Mix Optimization}}
\author{Joachim Gudmundsson%
	\and Pat Morin%
	\and Michiel Smid}

\newcommand{\cost}{\operatorname{cost}}
\newcommand{\ppu}{\operatorname{ppu}}
\newcommand{\val}{\operatorname{profit}}
\newcommand{\eps}{\epsilon}

\begin{document}
\maketitle
\begin{abstract}
  Algorithms for determining quality/cost/price tradeoffs in saturated
  markets are considered.  A product is modeled by $d$ real-valued
  qualities whose sum determines the unit cost of producing the product.
  This leads to the following optimization problem: given a set of $n$
  customers, each of whom has certain minimum quality requirements and a
  maximum price they are willing to pay, design a new product and select
  a price for that product in order to maximize the resulting profit.

  An $O(n\log n)$ time algorithm is given for the case, $d=1$, of linear
  products, and $O(n(\log n)^{d+1})$ time approximation algorithms
  are given for products with any constant number, $d$, of qualities.
  To achieve the latter result, an $O(nk^{d-1})$ bound on the complexity
  of an arrangement of homothetic simplices in $\R^d$ is given, where $k$
  is the maximum number of simplices that all contain a single points.
\end{abstract}

\section{Introduction}

Revealed preference theory \cite{v06} is a method of determining a course
of business action through the review of historical consumer behaviour.  In
particular, it is a method of inferring an individual's or a group's
preferences based on their past choices.  The \emph{marketing mix}
\cite{kpkl05} of a product consists of the 4 Ps: Product, price, place, and
promotion.  In the current paper, we present algorithms for optimizing the
first two of these by using data about consumer's preferences.  That is, we
show how, given data on consumer preferences, to efficiently choose a
product and a price for that product in order to maximize profit.

\begin{figure}
  \begin{center}
    \includegraphics{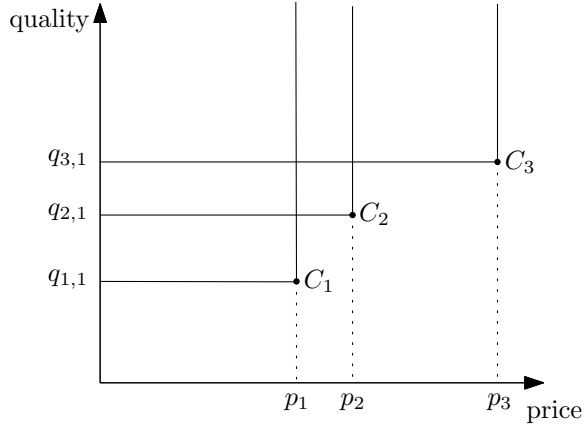}
  \end{center}
  \caption{A sample market with $d=1$ and $n=3$. A customer will consider
           any product that is in their upper left quadrant.}
  \figlabel{intro}
\end{figure}

Refer to \figref{intro}.  A product $P=(p,q_1,\ldots,q_d)$ is defined
by a real-valued \emph{price}, $p$, and a number of real-valued
orthogonal \emph{qualities}, $q_1,\ldots,q_d$.  The \emph{market} for
a product is a collection of customers $C=\{C_1,\ldots,C_n\}$, where
$C_i=(p_i,q_{i,1},\ldots,q_{i,d})$.  A customer will purchase the least
expensive product that meets all their minimum quality requirements and
whose price is below their maximum price.  That is, the customer $C_i$
will \emph{consider} the product $P=(p,q_1,\ldots,q_d)$ if $p \le p_i$
and $q_j \ge q_{i,j}$ for all $j\in\{1,\ldots,d\}$.  The customer $C_i$
will \emph{purchase} the product if has minimum price among all available
products that $C_i$ considers.

We consider markets that are \emph{saturated}.  That is, for every customer
$C_i$ there is an existing product that satisfies $C_i$'s
requirements and among all products that satisfy $C_i$'s requirements,
$C_i$ will choose the least expensive product.  From the point of view of a
manufacturer introducing one or more new products, this means that all
customers are \emph{Pareto optimal}, i.e., there are no two customers $C_i$
and $C_j$ such that $q_{i,k} > q_{j,k}$ for all $k\in\{1,\ldots,d\}$ and
$p_i < p_j$.  This is because, in a saturated market, $C_j$ and $C_i$ will
purchase the same product, namely the lowest-priced product that satisfies
$C_i$'s (and therefore also $C_j$'s) requirements.  When modelling a
saturated market, there is no need to explicitly consider existing products
since these can be encoded into the customers themselves.

As an example, consider a market for computers in which an example customer
$C_i$ may be looking for a computer with a minimum of 8 GB of RAM, a CPU
benchmark score of at least 3000, a GPU benchmark score of at least 2000,
and be willing to pay at most \$1500.  In addition, there is already a
computer on the market which meets these requirements and retails for
\$1200.  Thus, this customer would be described by the vector
$(1200,8,3000,2000)$.  If a manufacturer introduces a new product
$(1199,8,3500,2000)$ (a computer with 8 GB of RAM, a CPU benchmark score of
3500 and a GPU benchmark score of 2000 retailing for \$1199) then this
customer would select this new product over their current choice.

By appropriately reparameterizing the axes, we can assume that the cost,
$\cost(P)$, of manufacturing a product $P=(p,q_1,\ldots,q_d)$ is equal to
the sum of its qualities
\[
   \cost(P) = \sum_{i=1}^d q_i \enspace .
\]
The \emph{profit per unit sold} of $P$ is therefore
\[
   \ppu(P) = p-\cost(P) \enspace .
\]
In this paper we consider algorithms that a manufacturer can use when
a choosing new product to introduce into an existing saturated market
with the goal being to maximize the profit obtained.  More precisely,
given a Pareto-optimal \emph{market} of customers $M=\{C_1,\ldots,C_n\}$,
each have $d$ qualities, the $\textsc{ProductDesign}(d)$ problem is to
find a product $P^*\in\R^{d+1}$ such that
\[
  \val(P^*) = \ppu(P^*)
    \times 
      \left| \left\{ i:\mbox{$C_i$ purchases $P^*$} \right\} \right|
\]
is maximized.  To the best of our knowledge, this is the first time a
problem like this has been considered from an algorithmic perspective.

In the remainder of the paper we give an $O(n\log n)$ time algorithm for
$\textsc{ProductDesign}(1)$ (\secref{1-d}), and $O(n(\log n)^{d+1})$ time
approximation schemes for $\textsc{ProductDesign}(d)$ (\secref{2-d}
and \secref{d-d}).  \secref{conclusion} summarizes our results and
concludes with directions for future research.

\section{Linear products}
\seclabel{1-d}

In this section, we consider the simplest case, when a manufacturer
wishes to introduce a new product in which the quality of a product
has only one dimension.  Examples of such markets include, for example,
suppliers to the construction industry in which items (steel I-beams,
finished lumber, logs) must have a certain minimum length to be used for
a particular application.  An overly long piece can be cut down to size,
but using two short pieces instead of one long piece is not an option.

Throughout this section, since $d=1$, we will use the shorthand $P=(p,q)$
for the product being designed and $q_i$ for $q_{i,1}$.  Thus, we have a
set of customers $M=\{(p_1,q_1),\ldots,(p_{n},q_n)\}$ and we are searching
for a point $P^*=(p^*,q^*)$ that maximizes 
\[
   \val(p^*,q^*) = (p^*-q^*)
     |\{i : \mbox{$p^*\le p_i$ and $q^*\ge q_i$}\}|  \enspace .
\]

Our algorithm is an implementation of the \emph{plane-sweep} paradigm
\cite{bo79}. The correctness of the algorithm relies on two lemmas about
the structure of the solution space.  The first lemma is quite easy:
\begin{lem}\lemlabel{discrete}
  The value $(p^*,q^*)$ that maximizes $\val(p^*,q^*)$ is obtained when
  $p^* = p_i$ and $q^*=q_j$ for some $i,j\in\{1,\ldots,n\}$.
\end{lem}

\begin{proof}
  First, observe the obvious bounds on $p^*$ and $q^*$:
  \[
     \min\{p_i:i\in\{1,\ldots,n\}\} \le p^* 
      \le \max\{p_i:i\in\{1,\ldots,n\}\} 
  \] 
  and 
  \[
     \min\{q_i:i\in\{1,\ldots,n\}\} \le q^* 
      \le \max\{q_i:i\in\{1,\ldots,n\}\} \enspace .
  \] 
  Consider the arrangement of lines obtained by drawing a
  horizontal and vertical line through each customer $(p_i,q_i)$
  for $i\in\{1,\ldots,n\}$.  Within each cell of this arrangement,
  the function $\val(p,q)$ is a linear function of $p$ and $q$ and it is
  bounded.  Therefore, within a particular cell, the function is maximized
  at a vertex.  Since each vertex is the intersection of a horizontal
  and vertical line through a pair of customers, the lemma follows.
\end{proof}

The following lemma, illustrated in \figref{lemma-monotone}, is a
little more subtle and illustrates a manufacturer's preference for
lower-quality products:
\begin{figure}
  \begin{center}
    \includegraphics{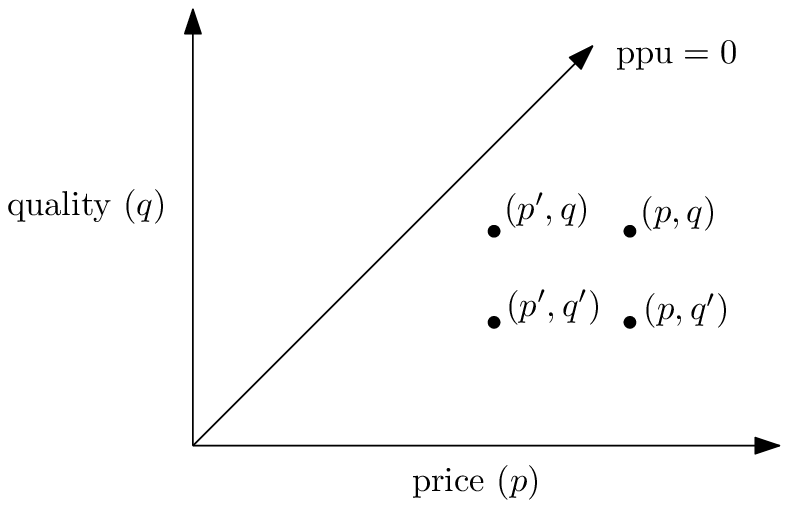}
  \end{center}
  \caption{$\val(p,q) \le \val(p,q')$ implies that $\val(p',q) \le
           \val(p',q')$ for all $p' \le p$.}
  \figlabel{lemma-monotone}
\end{figure}

\begin{lem}\lemlabel{monotone}
  Let $q' \le q$ and let $p$ be such that $0 < \val(p,q) \le \val(p,q')$.
  Then, for any $p' \le p$, $\val(p',q) \le \val(p',q')$.
\end{lem}

\begin{proof}
  By definition, $\val(p,q) = a(p-q)$ and $\val(p,q') = a'(p-q')$, where
  $a$ and $a'$ are the number of customers who would consider $(p,q)$
  and $(p,q')$, respectively.  These customers are all taken from the
  set $M_\ge =\{(p_i,q_i)\in M: p_i \ge p\}$.

  Now, consider the customers in the set $M'=\{(p_i,q_i)\in M: p' \le
  p_i < p\}$.  By the assumption that customers are Pareto optimal,
  any customer $(p_i,q_i)$ in $M'$ has $q_i \le q'$, so all of these
  customers will consider either $(p',q')$ or $(p',q)$ if either one
  is offered.  Therefore,
  \[
    \begin{aligned}
      \val(p',q')
        &  =   (a'+|M'|)(p'-q') \\
        &  =   a'(p'-q') + |M'|(p'-q') \\
        & \ge  a'(p'-q') + |M'|(p'-q) 
               && \mbox{since $q > q'$} \\
        &  =   a'(p-q') + a'(p'-p) + |M'|(p'-q) \\
        & \ge  a'(p-q') + a(p'-p) + |M'|(p'-q) 
               && \mbox{since $a \ge a'$ and $(p'-p) < 0$} \\
        & \ge  a(p-q) + a(p'-p) + |M'|(p'-q) 
               && \mbox{by assumption} \\
        &  =  a(p'-q) + |M'|(p'-q) \\
        &  =  \val(p',q) \enspace , \\
    \end{aligned}
  \]
  as required.
\end{proof}

\lemref{monotone} allows us to apply the plane sweep paradigm with a sweep
by decreasing price.  It tells us that, if a product $(p,q')$ gives better
profit than the higher-quality product $(p,q)$ at the current price $p$,
then it will always give a better profit for the remainder of the sweep.
In particular, there will never be a reason to consider a product with
quality $q$ for the remainder of the algorithm's execution.

Let the customers be labelled $(p_1,q_1),\ldots,(p_n,q_n)$ in decreasing
order of $p_i$, so that $p_{i+1} \le p_i$ for all $i\in\{1,\ldots,n-1\}$.
At any point in the sweep algorithm, there is a current price $p$, which
starts at $p=\infty$ and takes on the values $p_1,\ldots,p_n$,
successively, during the execution of the algorithm.  At all times, the
algorithm maintains a list $L$ of qualities $q_1^* > q_2^* > \cdots >
q_m^*$ such that $\val(p,q_1^*) > \val(p,q_2^*) >\cdots>\val(p,q_m^*)$.
The quality $q_1^*$ is the optimal quality for the current price, $p$.  By
the time the algorithm terminates, the quality of the globally-optimal
solution will have appeared as the first element in $L$.

To complete the description of the algorithm, all that remains is to
show how $L$ is updated during the processing of a sweep line event.
For this, the algorithm uses an auxiliary structure $D$ to efficiently
identify items in $L$ that need to updated.  Consider a consecutive
pair of the elements $q^*_i$ and $q^*_{i+1}$ in $L$.  When $q^*_i$ and
$q^*_{i+1}$ became adjacent in $L$, it was at some price $p=p_t$ such
that $\val(p_t,q^*_i)>\val(p_t,q^*_{i+1})$.  Let $a_i$ and $a_{i+1}$
be the number of customers who would consider $(p_t,q^*_i)$ and
$(p_t,q^*_{i+1})$, respectively. Then,
\[
  \val(p_t,q^*_i) = (p_t-q^*_i)a_i
\]
and
\[
  \val(p_t,q^*_{i+1}) = (p_t-q^*_{i+1})a_{i+1}
\]
Now, looking forward in time to a later step in the execution of
the algorithm, when $p=p_{t'}$, with $t'> t$, we find that
\[
  \val(p_{t'},q^*_i) = (p_{t'}-q^*_i)(a_i+t'-t)
\]
and
\[
  \val(p_{t'},q^*_{i+1}) = (p_{t'}-q^*_{i+1})(a_{i+1}+t'-t) \enspace .
\]
We are interested in identifying when the inequality
$\val(p_{t'},q^*_i) > \val(p_{t'},q^*_{i+1})$ changes to 
$\val(p_{t'},q^*_i) \le \val(p_{t'},q^*_{i+1})$.  That is, we need to
identify all indices $i$ for which, 
\begin{equation}
  (p_{t'}-q^*_i)(a_i+t'-t) \le (p_{t'}-q^*_{i+1})(a_{i+1}+t'-t)  \enspace ,
    \eqlabel{halfplane}
\end{equation}
at which point $q^*_i$ should be removed from $L$.  Observe that
the values of $a_i$, $a_{i+1}$, $q^*_i$, $q^*_{i+1}$, and $t$ are all
fixed at the time $q^*_i$ and $q^*_{i+1}$ become adjacent in $L$ and
the only values that change are those of $p_{t'}$ and $t'$.  Thus,
\eqref{halfplane} defines a halfplane in the plane with axes
$p_{t'}$ and $t'$.

The auxiliary data structure $D$ used by the algorithm must therefore be
able to store halfplanes and handle insertions, deletions, and queries
of the form ``Given a point $(x,y)$ return all halfplanes that contain
$(x,y)$.'' There are several data structures that solve this problem,
but the most suitable for the current application is the recent dynamic
convex hull data structure of Brodal and Jacob \cite{bj02}.\footnote{This
data structure is infamously complicated.  It's use in our application
can, however, be replaced with a much simpler semi-dynamic data structure
\cite{ds91}.  So as not to distract from the problem at hand, we defer
the discussion on how to do this until \secref{conclusion}.}  Their
data structure allows for the insertion and deletion of halfplanes in
$O(\log n)$ time.  Given a query point $(x,y)$, the data structure is
able to, in $O(\log n)$ time, find a single halfplane (if one exists)
that contains $(x,y)$.

The data structure $D$ is used as follows.  When the sweep line is
advanced to a new price $p_{t'}$, the value $q_{t'}$ is appended to
$L$ and the halfplane defined by $q_{t'}$ and its predecessor in $L$
is inserted into $D$.  Next, the data structure $D$ is repeatedly
queried with the point $(p_{t'},t')$.  This returns a halfplane
$h$ (if any exists) that corresponds to a pair of consecutive
elements $(q_i^*,q_{i+1}^*)$ such that $\val(p_{t'},q_i^*) \le
\val(p_{t'},q_{i+1}^*)$.  The halfplane $h$ is then deleted from $D$,
$q_{i}$ is deleted from $L$, and a new halfplane corresponding to the
(now adjacent) elements $q_{i-1}$ and $q_{i+1}$ is inserted into $D$.
This process is repeated until querying $D$ with the value $(p_{t'},t')$
returns no result.

Note that, after all the processing associated with updating the price
$p_{t'}$, the first element, $q_1$, in $L$ is the value that maximizes
$\val(q_1,p_{t'})$. Thus, the algorithm need only keep track, throughout
its execution, of the highest profit obtained from the first element of
$L$, and output this value at the end of its execution.  This completes
the description of the algorithm.

\begin{thm}\thmlabel{1-d}
  There exists an $O(n\log n)$ time algorithm for
  \textsc{ProductDesign$(1)$}.
\end{thm}

\begin{proof}
  The correctness of the algorithm described above follows from 2 facts:
  \lemref{discrete} ensures that the optimal solution is of the form
  $(p_i,q^*)$ for some $i\in\{1,\ldots,n\}$, and \lemref{monotone}
  ensures that the optimal solution appears at some point as the first
  element of the list $L$.

  The running time of the algorithm can be bounded as follows: Presorting
  the customers by decreasing order of price can be done in $O(n\log n)$
  time using any $O(n\log n)$ time sorting algorithm. Each sweep line
  event involves 1 insertion into $D$ plus some number $k$ of deletions,
  and insertions, and $k+1$ queries.  Note that each deletion in $D$
  corresponds to a deletion in $L$, and each element of $q_1,\ldots,q_n$
  is inserted into $L$ at most once.  Therefore, the total number of such
  deletions during the entire execution of the algorithm does not exceed
  $n$, and each such insertion/deletion pair takes $O(\log n)$ time.
  Since there are $n$ events, we conclude that the total running time
  of the algorithm is $O(n\log n)$, as claimed.
\end{proof}

The following theorem shows that a running time of $\Omega(n\log n)$
is inherent in this problem, even when considering approximation
algorithms.

\begin{thm}\thmlabel{1-d-lower-bound}
  Let $M$ be an instance of $\textsc{ProductDesign(1)}$ and $(p^*,q^*)$
  be a solution that maximizes $\val(p^*,q^*)$.  In the algebraic decision
  tree model of computation, any algorithm that can find a solution
  $(p,q)$ such that $2\cdot\val(p,q) > \val(p^*,q^*)$ has $\Omega(n\log n)$
  running time in the worst-case.
\end{thm}

\begin{proof}
  We reduce from the integer \textsc{Element-Uniqueness} problem, which
  has an $\Omega(n\log n)$ lower bound in the algebraic decision tree
  model \cite{y91}: Given an array $A=[x_1,\ldots,x_n]$ containing
  $n$ integers, are all the elements of $A$ unique?

  We convert $A$ into an instance of $\textsc{ProductDesign}(1)$ in
  $O(n)$ time as follows (refer to \figref{lemma-lower-bound}).  For each
  $x_i$, $i\in\{1,\ldots,n\}$ we introduce a customer $(p_i,q_i)$ with
  $p_i=q_i+1/2$ and $q_i=x_i$.  If there exists a value $x$ in $A$ that
  occurs $2$ or more times, then the product $(x+1/2,x)$ gives a value
  $\val(x+1/2,x) \ge 1$.  On the other hand, if there is no such $x$,
  then 
  \begin{enumerate} 
    \item any product $(p,q)$ with $p-q>1/2$ can not be sold to any
    customers and
    \item any product $(p,q)$ with $p-q>0$ can be sold to at most $1$
    customer.
  \end{enumerate} 
  Therefore, if all the elements of $A$ are unique, then $\val(p^*,q^*)
  = 1/2$, otherwise $\val(p^*,q^*) \ge 1$.  The result follows.
\end{proof}

\begin{figure}
  \begin{center}
    \includegraphics{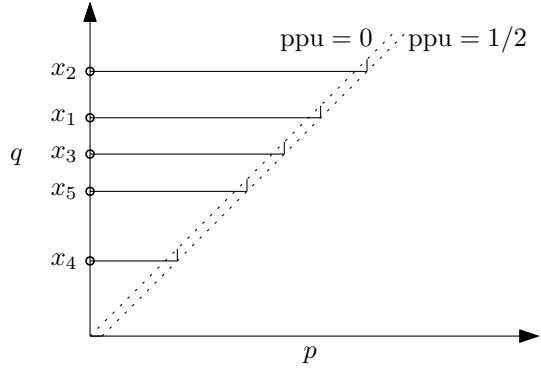}
  \end{center}
  \caption{Reducing \textsc{Element-Uniqueness} to
           $\textsc{ProductDesign(1,1)}$.}
  \figlabel{lemma-lower-bound}
\end{figure}

\section{A near-linear approximation algorithm for bidimensional products}
\seclabel{2-d}

In this section, we consider algorithms for $\textsc{ProductDesign}(2)$, in
which products have 2 qualities.  As a baseline, we first observe that, if
we fix the value of $q_2$, then the optimal solution of the form $(p, q_1,
q_2)$ can be found using a single application of the algorithm in
\thmref{1-d}.  Therefore, by successively solving the problem for each $q_2
\in\{q_{2,1},\ldots,q_{2,n}\}$ and taking the best overall solution we
obtain an $O(n^2\log n)$ time algorithm for $\textsc{ProductDesign}(2)$.

More generally, $\textsc{ProductDesign}(d)$ can be solved using
$O(n^{d-1})$ applications of \thmref{1-d} resulting in an $O(n^d\log n)$
time algorithm.  Unfortunately, these are the best results known for $d \ge
2$, and, as discussed in \secref{conclusion}, we suspect that an algorithm
with running time $o(n^d)$ will be difficult to achieve using existing
techniques.  Therefore, in this section we focus our efforts on obtaining a
near-linear approximation algorithm.

Fix some constant $\epsilon > 0$.  Given an instance $M$ of
$\textsc{ProductDesign}(d)$, a point $P\in\R^{d+1}$ is a
$(1-\eps)$-approximate solution for $M$ if $\val(P) \ge (1-\eps)\val(P^*)$
for all $P^*\in\R^{d+1}$.  An algorithm is a (high probability)
\emph{Monte-Carlo $(1-\eps)$-approximation algorithm} for
$\textsc{ProductDesign}(d)$ if, given an instance $M$ of size $n$, the
algorithm outputs a $(1-\eps)$-approximate solution for $M$ with
probability at least $1-n^{-c}$ for some constant $c>0$.

Let $r=\max\{\ppu(C_i) : i\in\{1,\ldots, n\}\}$ and observe that $r$
is the maximum profit per unit that can be achieved in this market.
Let $E=1/(1-\eps)$ and let $\ell = \lceil\log_E n\rceil$ and observe
that $\ell = O(\eps^{-1}\log n)$.\footnote{This can be seen by taking the
limit $\lim_{\eps\rightarrow 0^+} (\eps/\log(E))$ using one application
of L'H\^opital's Rule.} For each $i\in\{0,1,2,\ldots,\ell\}$, define the
plane $H_i = \{ (p,q_1,q_2) : p-q_1-q_2 = r(1-\eps)^i \}$.  The following
lemma says that a search for an approximate solution can be restricted
to be contained in one of the planes $H_i$.

\begin{lem}\lemlabel{plane-approx}
  For any product $P^*=(p^*,q_1^*,q_2^*)$, there exists a product
  $P=(p,q_1,q_2)$ such that $P\in H_i$ for some $i\in\{0,\ldots,\ell\}$
  and $\val(P) \ge (1-\eps)\val(P^*)$.
\end{lem}

\begin{proof}
  There are two cases to consider.  If $\ppu(P^*) \le r/n$ then $\val(P^*)
  \le r$, in which case we set $P=C_i$ where $\ppu(C_i) = r$, so that
  $P\in H_0$ and $\val(P) = r \ge \val(P^*)\ge (1-\eps)\val(P^*)$,
  as required.

  Otherwise, $r/n < \ppu(P^*) \le r$.  In this case, consider the
  plane $H_i$ where $i = \lceil\log_E (r/\ppu(P^*))\rceil$.  Notice,
  that for any point $P\in H_i$, $\ppu(P) \ge (1-\eps)\ppu(P^*)$.
  More specifically, the orthogonal projection $P=(p,q_1,q_2)$ of $P^*$
  onto $H_i$ is a product with $p\le p^*$, $q_1\ge q_1^*$, and $q_2\ge
  q_2^*$.  Therefore, any customer who would consider $P^*$ would also
  consider $P$, so $\val(P) \ge (1-\eps)\val(P^*)$, as required.
\end{proof}

\lemref{plane-approx} implies that the problem of finding an approximate
solution to  $\textsc{ProductDesign}(2)$ can be reduced to a sequence of
problems on the planes $H_0,\ldots,H_\ell$.  Refer to \figref{plane}. Each
customer $C_j$ considers all products in a quadrant whose corner is $C_j$.
The intersection of this quadrant with $H_i$ is a (possibly empty)
equilateral triangle $\Delta_{i,j}$.  The customer $C_j$ will consider
a product $P$ in $H_i$ if and only $P$ is in $\Delta_{i,j}$.  Thus,
the problem of solving $\textsc{ProductDesign}(2)$ restricted to the
plane $H_i$ is the problem of finding a point contained in the largest
number of equilateral triangles from the set $\Delta_i=\{\Delta_{i,j}:
j\in\{1,\ldots,n\}\}$.

\begin{figure}
  \begin{center}
    \includegraphics[width=\textwidth]{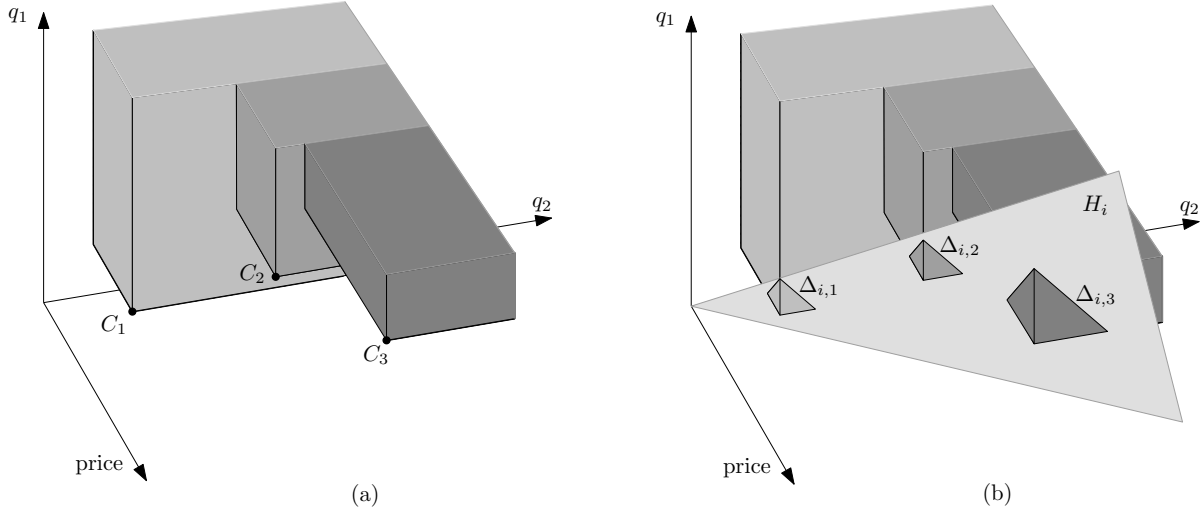}
  \end{center}
  \caption{The intersection of $H_i$ with customers' quadrants is a
          set of homothetic equilateral triangles.}
  \figlabel{plane}
\end{figure}

Note that the elements in $\Delta_i$ are \emph{homothets} (translations and
scalings) of an equilateral triangle, so they form a collection of
\emph{pseudodisks} and we wish to find the deepest point in this collection
of pseudodisks.  No algorithm with running time $o(n^2)$ is known for
solving this problem exactly, but Aronov and Har-Peled \cite{ah08} have
recently given a Monte-Carlo $(1-\eps)$-approximation algorithm for this
problem that runs in time $O(\eps^{-2}n\log n)$.  By applying this
algorithm to each of $\Delta_i$ for $i\in\{1,\ldots,\ell\}$, we obtain the
following result:

\begin{thm}
  For any $\eps >0$, there exists an $O(\eps^{-3}n(\log n)^2)$ time
  (high-probability) Monte-Carlo $(1-\eps)$-approximation algorithm for
  $\textsc{ProductDesign}(2)$.
\end{thm}

\section{A near-linear approximation algorithm for constant $d$}
\seclabel{d-d}

In this section we extend the algorithm from the previous section to
(approximately) solve $\textsc{ProductDesign}(d)$ for any constant value of
$d$.  The algorithm is more or less unchanged, except that the proof
requires some new results on the combinatorics of arrangements of
homothets.

As before, let $r=\max\{\ppu(C_i) : i\in\{1,\ldots, n\}\}$ and let
$\ell = \lceil\log_E n\rceil$. For each $i\in\{0,1,2,\ldots,\ell\}$,
define the hyperplane $H_i = \{ (p,q_1,\ldots,q_d) : p-\sum_{i=1}^d
q_i =  r(1-\eps)^i \}$.  The following lemma has exactly the same proof
as \lemref{plane-approx}.

\begin{lem}\lemlabel{plane-approx-d}
  For any product $P^*=(p^*,q_1^*,\ldots,q_d^*)$, there exists
  a product $P=(p,q_1,\ldots,q_d)$ such that $P\in H_i$ for some
  $i\in\{0,\ldots,\ell\}$ and $\val(P) \ge (1-\eps)\val(P^*)$.
\end{lem}

Again, each customer $C_j$ defines a regular simplex $\Delta_{i,j}$
in $H_{i}$ such that $C_j$ will consider $P\in H_i$ if and
only if $P\in\Delta_{i,j}$.  In this way, we obtain a set
$\Delta_i=\{\Delta_{i,1},\ldots,\Delta_{i,n}\}$ of homothets of
a regular simplex in $\R^d$ and we require an algorithm to find a
($(1-\eps)$-approximation to) the point that is contained in the largest
number of these simplices.  The machinery of Aronov and Har-Peled
\cite{ah08} can be used to help solve this problem, but not before we
prove some preliminary results, the first of which is a combinatorial
geometry result.

\begin{lem}\lemlabel{arrangement}
  Let $\Delta$ be a set of $n$ homothets of a regular simplex in $\R^d$,
  for $d=O(1)$, and such that no point in $\R^d$ is contained in more
  than $k$ elements of $\Delta$.  Then, the total complexity of the
  arrangement, $A(\Delta)$, of the simplices in $\Delta$ is $O(nk^{d-1})$.
\end{lem}

\begin{proof}
  We first consider the simpler case in which the elements of $\Delta$
  are translates (without scaling) of a regular simplex.  Suppose that
  the total complexity of $A(\Delta)$ is $m$.  Then, by the pigeonhole
  principle, there is some element $T$ in $\Delta$ that is involved
  in $m/n$ features of $A(\Delta)$.  Note that this implies that $T$
  intersects all the elements of a set $\Delta'\subseteq\Delta$ with
  $|\Delta'|=\Omega((m/n)^{1/(d-1)})$, since otherwise there are not
  enough elements in $\Delta'$ to generate $m/n$ features on the surface
  of $\Delta$.

  Observe that, since the elements of $\Delta'$ are all unit size and
  they all intersect $T$, that they are all contained in a ball of
  radius $O(1)$ centered at the center of $T$.  Furthermore, since each
  element of $\Delta'$ has volume $\Omega(1)$ this implies that some point
  must be contained in $\Omega((m/n)^{1/(d-1)})$ elements of $\Delta'$.
  Thus, we obtain the inequality $k \ge \Omega((m/n)^{1/(d-1)})$, or,
  equivalently, $m \le O(nk^{d-1})$, as required.

  Now, for the case where the elements of $\Delta$ are homothets
  (translations and scalings) of a regular simplex, we proceed as follows.
  Assume, by way of contradiction, that $|A(\Delta)| > rn$ for some $r$
  to be defined later.  Label the elements of $\Delta$ $T_1,\ldots,T_n$
  in increasing order of size and consider the smallest element $T_i$ such
  that $T_i$ contributes at least $r$ features to $A(\{T_i,\ldots,T_n\})$.
  Such a $T_i$ is guaranteed to exist, since otherwise $|A(\Delta)|\le
  rn$.

  Now, $T_i$ intersects all the elements in some set $\Delta'\subseteq
  \{T_{i+1},\ldots,T_n\}$ with $|\Delta'| = \Omega(r^{1/(d-1)})$.
  Shrink each element $T'$ in $\Delta'$ so as to obtain an element
  $T''$ such that (a)~the size of $T''$ is equal to the size of $T$ and
  (b)~$T'' \subseteq T'$.  Call the resulting set of shrunken elements
  $\Delta''$.  Condition~(a) and the packing argument above imply that
  there is a point $p\in\R^d$ that is contained in $\Omega(r^{1/(d-1)})$
  elements of $\Delta''$.  Condition~(b) implies that $p$ is contained in
  $\Omega(r^{1/(d-1)})$ elements of $\Delta'$ and hence also $\Delta$.
  Therefore, we conclude, as before that $r \le O(k^{d-1})$.  Thus,
  for a sufficiently large constant $c$, setting $r=ck^{d-1}$ yields a
  contradiction to the assumption that $|A(C)| > rn$.  We conclude that
  $|A(C)| = O(nk^{d-1})$, as required.
\end{proof}

\paragraph{Remark.} The proof of \lemref{arrangement} makes almost no use
of the assumption that the elements of $\Delta$ are regular simplices other
than using the property that their volume is related to their diameter.
Thus, a version of this lemma holds for collections of \emph{fat} objects,
a result that will probably be of independent interest.

\paragraph{Remark.} \lemref{arrangement} is somewhat surprising, since the
union of $n$ homothets of a regular tetrahedron in, for example, $\R^3$ can
easily have complexity $\Omega(n^2)$.  This fact makes it impossible to
apply the ``usual'' Clarkson-Shor technique \cite{cs89} that relates the
complexity of the first $k$ levels to that of the boundary of the union
(the 0-level).

\begin{lem}\lemlabel{arrangement-construct}
  Let $\Delta$ be a set of $n$ homothets of the regular simplex in
  $\R^d$ such that no point of $\R^d$ is contained in more than $k$
  simplices of $\Delta$.  Then the arrangement $A(\Delta)$ of $\Delta$
  can be computed in $O(n(k^{d-1}+(\log n)^{d}))$ time.
\end{lem}

\begin{proof}
  Computing the arrangement $A(\Delta)$ can be done in the following way.
  Sort the elements of $\Delta$ by decreasing size and construct $A(C)$
  incrementally by inserting the elements one by one.  When inserting
  an element $T$, use a data structure (described below) to retrieve
  the elements of $\Delta$ that intersect $T$ and discard the elements
  that are smaller than $T$.  The proof of \lemref{arrangement} implies
  that there will be at most $O(k)$ such elements.  The intersection
  of the surfaces of these $O(k)$ elements with the surface of $T$ has
  size $O(k^{d-1})$ and can be computed in $O(k^{d-1})$ time using $d+1$
  applications of the standard algorithm for computing an arrangement
  of hyperplanes in $\R^{d-1}$ \cite{eos86,ess93}.  Thus, ignoring the
  cost of finding the elements that intersect $T$, the overall running
  time of the algorithm is $O(nk^{d-1})$.

  All that remains is to describe a data structure for retrieving the
  elements that intersect a given simplex $T\in\Delta$.  In the following
  we describe a data structure that can be constructed in $O(n(\log
  n)^{d})$ time and can answer queries in $O(x + (\log n)^{d})$ time,
  where $x$ is the size of the output.  This data structure will be
  constructed once and queried $n$ times.  The total size of the outputs
  over all $n$ queries will be the $O(|A(C)|)=O(nk^{d-1})$.

  Since the elements of $\Delta$ are homothets, each element of $\Delta$
  can be described concisely by the lexicographically smallest of its
  coordinates and one real number describing its size. In this way,
  each element of $\Delta$ can be represented as a point in $\R^{d+1}$
  and thus $\Delta$ can be represented as a set $\Delta'\subseteq\R^{d+1}$
  of $n$ points .

  Observe that, given a half space $h$ in $\R^d$, the set of elements of
  $\Delta$ intersected by $h$ is given by the set of points of $\Delta'$
  contained in some halfspace $h'=f(h)$ in $\R^{d+1}$.  Furthermore,
  the transformation $f$ preserves parallel halfspaces.  That is,
  if $h_1,h_2\subseteq{\R^d}$ are halfspaces that have the same inner
  normal then $f(h_1)$ and $f(h_2)$ also have the same inner normal.

  Now, observe that any simplex $T\in\Delta$ can be expressed as the
  intersection of $d$ closed halfspaces, all of whose inner normals are
  taken from the set of $N=\{n_1,\ldots,n_{d+1}\}\subset \Sp^{d-1}$ of
  inner normals of the facets of a regular simplex.  Let $f(n_i)\in\Sp^d$
  denote the inner normal of a halfspace $f(h)$, where $h\subset\R^d$
  is a halfspace whose inner normal is $n_i$.

  The data structure for storing $\Delta'$ is a $d+1$ layer range tree
  \cite{b75} in which the $i$th layer, for $i\in\{1,\ldots,d+1\}$,
  stores its points ordered by their projection onto $f(n_i)$.  In this
  way, the range tree can return the set of all simplices in $\Delta$ that
  intersect $T$.  The size of this range tree is $O(n(\log n)^{d})$
  and it can answer queries in time $O(x+n(\log n)^{d})$ where $x$
  is the size of the query result.  Since each simplex in $\Delta$
  is passed as a query to this data structure exactly once, the total
  sizes of outputs over all $n$ queries is equal to the number of
  pairs $T_1,\ldots,T_2\in\Delta$ such that $T_1\cap T_2\neq\emptyset$.
  But the number of such pairs is certainly a lower bound on $|A(\Delta)|$
  so it must be at most $O(nk^{d-1})$.  This completes the proof.
\end{proof}

\lemref{arrangement-construct} can be used as a subroutine in the
algorithm of Aronov and Har-Peled \cite[Theorem~3.3]{ah08}, to obtain
the following Corollary.

\begin{cor}\corlabel{deep-point}
  Let $\Delta$ be a set of $n$ homothets of a regular simplex in $\R^d$
  such that some point $p\in\R^d$ is contained in $\delta$ elements
  of $\Delta$.  Then there exists an algorithm whose running time is
  $O(\eps^{-2d}n(\log n)^{d-1} + n(\log n)^d)$ and that, with high
  probability, returns a point $p'\in\R^{d}$ contained in at least
  $(1-\eps)\delta$ elements of $\Delta$.
\end{cor}

As before, an approximate solution to $\textsc{ProductDesign}(d)$ problem
reduces to finding deepest point in each of the sets
$\Delta_1,\ldots,\Delta_\ell$ where $\Delta_i$ is a set of $n$
$d$-simplices in $H_i$.  By using the algorithm of \corref{deep-point} to
do this we obtain the following result:

\begin{thm}
  For any $\eps >0$, there exists an $O(\eps^{-(2d+1)}n(\log n)^d + n(\log
  n)^{d+1})$ time (high-probability) Monte-Carlo $(1-\eps)$-approximation
  algorithm for $\textsc{ProductDesign}(d)$.
\end{thm}

\section{Conclusions}
\seclabel{conclusion}

We have given an $O(n\log n)$ time exact algorithm for solving
$\textsc{ProductDesign}(1)$ and $O(n(\log n)^{d+1})$ time approximation
algorithms for solving $\textsc{ProductDesign}(d)$.  The running time of
the exact $\textsc{ProductDesign}(1)$ algorithm is optimal and no algorithm
that produces a $(2-\eps)$-approximation, for any $\epsilon > 0$, can run
in $o(n\log n)$ time.  

In developing these algorithms, we gave a proof (the
proof of \lemref{arrangement}) that shows that an arrangement of $n$ fat convex
objects in $\R^d$ has complexity $O(nk^{d-1})$ where $k$ is the maximum
number of objects that contain any given point. We expect that this result,
and the algorithm for approximate depth that arise from it \cite{ah08},
will find other applications.

Our exact algorithm for the case $d=2$ uses the recent dynamic convex
hull data structure of Brodal and Jacob \cite{bj02}, which is quite
complicated.  We observe that the use of this structure can be avoided
by using the (much simpler) semi-dynamic data structure of Dobkin and Suri
\cite{ds91}.  To use this structure, we need to specify, each time a point is
inserted, the time at which that point will be deleted.  In our case,
\lemref{monotone} implies that this deletion time can be computed by a
binary search on $p_t,\ldots,p_n$ using \eqref{halfplane}.

An exact near-linear time algorithm for the case $d=2$ seems to be out of
reach.  It appears as if this problem requires (at least) a solution to the
problem of finding a point contained in the largest number of homothets of
an equilateral triangle, a problem for which no subquadratic time algorithm
is known.  Is it possible to prove some kind of a lower bound?  The related
problem of finding the point contained in the largest number of unit disks
is \textsc{3-Sum} hard \cite{ah08} providing some evidence that this
problem will be difficult to solve in subquadratic time.

In this paper we considered the case where the problem is parameterized by
the number, $d$, of orthogonal qualities that a product may have.  Another
case to consider is the case in which a manufacturer wishes to introduce
some number, $k$, $k>1$, of new products into a market.   Is this problem
NP-hard?  Does it have a polynomial time approximation algorithm?

\section*{Acknowledgement}

The authors would like to thank Gautam~Das for bringing this class of
problems to our attention and Timothy~Chan for helpful discussions on
the subject of approximate depth.

\bibliographystyle{plain}
\bibliography{pricing}

\end{document}